\documentclass[12pt,a4paper]{amsart}%
\usepackage{pict2e}
\usepackage{amsfonts}
\usepackage{amsthm}
\usepackage{amsmath}
\usepackage{amscd}
\usepackage[latin2]{inputenc}
\usepackage{t1enc}
\usepackage[mathscr]{eucal}
\usepackage{indentfirst}
\usepackage{graphicx}
\usepackage{graphics}
\usepackage[margin=2.9cm]{geometry}
\usepackage{tikzexternal}
\usepackage{amssymb}%
\usepackage{url}
\providecommand{\U}[1]{\protect\rule{.1in}{.1in}}
\providecommand{\U}[1]{\protect\rule{.1in}{.1in}}
\numberwithin{equation}{section}
\theoremstyle{plain}
\newtheorem{Th}{Theorem}[section]
\newtheorem{Lemma}[Th]{Lemma}

\newtheorem{Prop}[Th]{Proposition}
\theoremstyle{definition}
\newtheorem{Def}[Th]{Definition}

\newtheorem{Rem}[Th]{Remark}
\newtheorem{?}[Th]{Problem}
\newtheorem{Ex}[Th]{Example}

\begin{document}
\title[Matching measure and the monomer-dimer free energy]{Matching measure, Benjamini--Schramm convergence and the monomer-dimer free energy}
\author[M. Ab\'ert]{Mikl\'os Ab\'ert}
\address{Alfr\'ed R\'enyi Institute of Mathematics \\
H-1053 Budapest \\
Re\'altanoda u. 13-15. \\
Hungary}
\email{karinthy@gmail.com}
\author[P. Csikv\'ari]{P\'eter Csikv\'ari}
\address{Massachusetts Institute of Technology \\
Department of Mathematics \\
Cambridge MA 02139 \& E\"otv\"os Lor\'and University \\
Department of Computer Science \\
H-1117 Budapest \\
P\'azm\'any P\'eter s\'et\'any 1/C \\
Hungary}
\email{peter.csikvari@gmail.com}
\author[T. Hubai]{Tam\'as Hubai}
\address{Alfr\'ed R\'enyi Institute of Mathematics \\
H-1053 Budapest \\
Re\'altanoda u. 13-15. \\
Hungary \& E\"otv\"os Lor\'and University \\
Department of Computer Science \\
H-1117 Budapest \\
P\'azm\'any P\'eter s\'et\'any 1/C \\
Hungary}
\email{htamas@cs.elte.hu}
\thanks{The second author is partially supported by the Hungarian National Foundation
for Scientific Research (OTKA), grant no. K81310. All authors are partially
supported by MTA R\'enyi "Lend\"ulet" Groups and Graphs Research Group.}

\begin{abstract}
We define the \emph{matching measure} of a lattice $L$ as the spectral measure
of the tree of self-avoiding walks in $L$. We connect this invariant to the
monomer-dimer partition function of a sequence of finite graphs converging to
$L$.

This allows us to express the monomer-dimer free energy of $L$ in terms of the
matching measure. Exploiting an analytic advantage of the matching measure
over the Mayer series then leads to new, rigorous bounds on the monomer-dimer
free energies of various Euclidean lattices. While our estimates use only the
computational data given in previous papers, they improve the known bounds significantly.

\end{abstract}
\maketitle

\section{Introduction}

The aim of this paper is to define the matching measure of an infinite lattice
$L$ and show how it can be used to analyze the behaviour of the monomer-dimer
model on $L$. The notion of matching measure has been recently introduced by
the first and second authors, Frenkel and Kun in \cite{ACFK}. There are
essentially two ways to define it: in this paper we take the path of giving a
direct, spectral definition for infinite vertex transitive lattices, using
self-avoiding walks and then connect it to the monomer-dimer model via graph
convergence. Recall that a graph $L$ is \emph{vertex transitive} if for any
two vertices of $L$ there exists an automorphism of $L$ that brings one vertex
to the other.

Let $v$ be a fixed vertex of the graph $L$. A walk in $L$ is
\emph{self-avoiding}, if it touches every vertex at most once. There is a
natural graph structure on the set of finite self-avoiding walks starting at
$v$: we connect two walks if one is a one step extension of the other. The
resulting graph is an infinite rooted tree, called the \emph{tree of
self-avoiding walks} \emph{of }$L$\emph{\ starting at }$v$.

\begin{Def}
\label{saw}Let $L$ be an infinite vertex transitive lattice. The
\emph{matching measure} $\rho_{L}$ is the spectral measure of the tree of
self-avoiding walks of $L$ starting at $v$, where $v$ is a vertex of $L$.
\end{Def}

By vertex transitivity, the definition is independent of $v$. For a more
general definition, also covering lattices that are not vertex transitive, see
Section~\ref{measure}.

To make sense of why we call this the matching measure, we need the notion of
Benjamini--Schramm convergence. Let $G_{n}$ be a sequence of finite graphs. We
say that $G_{n}$ Benjamini--Schramm converges to $L$, if for every $R>0$, the
probability that the $R$-ball centered at a uniform random vertex of $G_{n}$
is isomorphic to the $R$-ball in $L$ tends to $1$ as $n$ tends to infinity.
That is, if by randomly sampling $G_{n}$\ and looking at a bounded distance,
we can not distinguish it from $L$ in probability.

All Euclidean lattices $L$ can be approximated this way by taking sequences of
boxes with side lengths tending to infinity, by bigger and bigger balls in $L$
in its graph metric, or by suitable tori. When $L$ is a Bethe lattice (a
$d$-regular tree), finite subgraphs never converge to $L$ and the usual way is
to set $G_{n}$ to be $d$-regular finite graphs where the minimal cycle length
tends to infinity.

For a finite graph $G$ and $k>0$ let $m_{k}(G)$ be the number of monomer-dimer
arrangements with $k$ dimers (matchings of $G$ using $k$ edges). Let
$m_{0}(G)=1$. Let the \emph{matching polynomial}
\[
\mu(G,x)=\sum_{k}(-1)^{k}m_{k}(G)x^{\left\vert G\right\vert -2k}%
\]
and let $\rho_{G}$, the \emph{matching measure of }$G$ be the uniform
distribution on the roots of $\mu(G,x)$. Note that $\mu(G,x)$ is just a
reparametrization of the monomer-dimer partition function. The matching
polynomial has the advantage over the  partition function that its roots are
bounded in terms of the maximal degree of $G$.

Using previous work of Godsil \cite{god3} we show that $\rho_{L}$ can be
obtained as the thermodynamical limit of the $\rho_{G_{n}}$.

\begin{Th}
\label{konv}Let $L$ be an infinite vertex transitive lattice and let $G_{n}$
Benjamini--Schramm converge to $L$. Then $\rho_{G_{n}}$ weakly converges to
$\rho_{L}$ and $\lim_{n\rightarrow\infty}\rho_{G_{n}}(\{x\})=\rho_{L}(\{x\})$
for all $x\in\mathbb{R}$.
\end{Th}

So in this sense, the matching measure can be thought of as the `root
distribution of the partition function for the infinite monomer-dimer model',
transformed by a fixed reparametrization.

It turns out that the matching measure can be effectively used as a substitute
for the Mayer series. An important advantage over it is that certain natural
functions can be integrated along this measure even in those cases when the
corresponding series do not converge. We demonstrate this advantage by giving
new, strong estimates on the free energies of monomer-dimer models for
Euclidean lattices, by expressing them directly from the matching measures.

The computation of monomer-dimer and dimer free energies has a long history.
The precise value is known only in very special cases. Such an exceptional
case is the Fisher-Kasteleyn-Temperley formula \cite{fis,kas,tem} for the
dimer model on $\mathbb{Z}^{2}$. There is no such exact result for
monomer-dimer models. The first approach for getting estimates was the use of the
transfer matrix method. Hammersley \cite{ham1,ham2}, Hammersley and Menon
\cite{ham3} and Baxter \cite{bax} obtained the first (non-rigorous) estimates
for the free energy. Then Friedland and Peled \cite{FP} proved the rigorous
estimates $0.6627989727\pm10^{-10}$ for $d=2$ and the range $[0.7653,0.7863]$
for $d=3$. Here the upper bounds were obtained by the transfer matrix method,
while the lower bounds relied on the Friedland-Tverberg inequality. The lower
bound in the Friedland-Peled paper was subsequently improved by newer and
newer results (see e.g. \cite{friegurv}) on Friedland's asymptotic matching
conjecture which was finally proved by L. Gurvits \cite{gur2}. Meanwhile, a
non-rigorous estimate $[0.7833,0.7861]$ was obtained via matrix permanents
\cite{huo}. The most significant improvement was obtained recently by D.
Gamarnik and D. Katz \cite{gam} via their new method which they called
sequential cavity method. They obtained the range $[0.78595,0.78599]$.
\bigskip

Here we only highlight one computational result. More data can be found in
Section~\ref{entropy-function}, in particular, in Table 1. Let $\tilde
{\lambda}(L)$ denote the monomer-dimer free energy of the lattice $L$, and let
$\mathbb{Z}^{d}$ denote the $d$-dimensional hyper-simple cubic lattice.

\begin{Th}
We have
\[
\tilde{\lambda}(\newline\mathbb{Z}^{3})=0.7859659243\pm9.88\cdot
10^{-7},
\]%
\[
\tilde{\lambda}(\newline\mathbb{Z}^{4})=0.8807178880\pm5.92\cdot
10^{-6}.
\]
\[
\tilde{\lambda}(\newline\mathbb{Z}^{5})=0.9581235802\pm4.02\cdot
10^{-5}.
\]
The bounds on the error terms are rigorous.
\end{Th}

Our method allows to get efficient estimates on arbitrary lattices. The
computational bottleneck is the tree of self-avoiding walks, which is famous
to withstand theoretical interrogation.

It is natural to ask what are the actual matching measures for the various
lattices. In the case of a Bethe lattice $\mathbb{T}_{d}$, the tree of
self-avoiding walks again equals $\mathbb{T}_{d}$, so the matching measure of
$\mathbb{T}_{d}$ coincides with its spectral measure. This explicit measure,
called Kesten-McKay measure has density
\[
\frac{d}{2\pi}\frac{\sqrt{4(d-1)-t^{2}}}{d^{2}-t^{2}}\chi_{\{|t|\leq
2\sqrt{d-1}\}}\text{.}%
\]

We were not able to find such explicit formulae for any of the Euclidean
lattices. However, using Theorem~\ref{konv} one can show that the matching
measures of hypersimple cubic lattices admit no atoms.

\begin{Th}
\label{atomless}The matching measures $\rho_{\mathbb{Z}^{d}}$ have no atoms.
\end{Th}

In Section~\ref{density} we prove a more general result which also shows that
for instance, the matching measure of the hexagonal lattice has no atoms. For
some images on the matching measures of $\mathbb{Z}^{2}$ and $\mathbb{Z}^{3}$
see Section~\ref{density}. We expect that the matching measures of all
hypersimple cubic lattices are absolutely continuous with respect to the
Lebesque measure. We also expect that the radius of support of the matching
measure (that is, the spectral radius of the tree of self-avoiding walks)
carries further interesting information about the lattice. Note that the
\emph{growth} of this tree for $\mathbb{Z}^{d}$ and other lattices has been under intense investigation
\cite{alm,dum,HSS}, under the name \emph{connective constant}. \bigskip

The paper is organized as follows. In Section~\ref{measure}, we define the
basic notions and prove Theorem~\ref{konv}. In Section~\ref{entropy-function} we
introduce the entropy function $\lambda_{G}(p)$ for finite graphs $G$ and
related functions, and we gather their most important properties. We also
extend this concept to lattices. In this section we provide the computational
data too. In Section~\ref{density}, we prove Theorem~\ref{atomless}.

\section{Matching measure}

\label{measure}

\subsection{Notations}

This section is about the basic notions and lemmas needed later. Since the
same objects have different names in graph theory and statistical mechanics,
for the convenience of the reader, we start with a short dictionary. \bigskip

\begin{center}%
\begin{tabular}
[c]{|c|c|}\hline
Graph theory & Statistical mechanics\\\hline
vertex & site\\\hline
edge & bond\\\hline
$k$-matching & monomer-dimer arrangement with $k$ dimers\\\hline
perfect matching & dimer arrangement\\\hline
degree & coordination number\\\hline
$d$-dimensional grid ($\mathbb{Z}^{d}$) & hyper-simple cubic lattice\\\hline
infinite $d$-regular tree ($\mathbb{T}_{d}$) & Bethe lattice\\\hline
path & self-avoiding walk\\\hline
\end{tabular}
\bigskip
\end{center}

Throughout the paper, $G$ denotes a finite graph with vertex set $V(G)$ and
edge set $E(G)$. The number of vertices is denoted by $|G|$. For an infinite
graph $L$, we will use the word \emph{lattice}. The \emph{degree} of a vertex
is the number of its neighbors. A graph is called $d$\emph{-regular} if every
vertex has degree exactly $d$. The graph $G-v$ denotes the graph obtained from
$G$ by erasing the vertex $v$ together with all edges incident to $v$.

For a finite or infinite graph $T$, let $l^{2}(T)$ denote the Hilbert space of
square summable real functions on $V(T)$. The \emph{adjacency operator}
$A_{T}:l^{2}(T)\rightarrow l^{2}(T)$ is defined by
\[
(A_{T}f)(x)=\sum\limits_{(x,y)\in E(T)}f(y)\text{ \ \ \ (}f\in l^{2}%
(T)\text{).}%
\]
When $T$ is finite, in the standard base of vertices, $A_{T}$ is a square
matrix, where $a_{u,v}=1$ if the vertices $u$ and $v$ are adjacent, otherwise
$a_{u,v}=0$. For a finite graph $T$, the characteristic polynomial of $A_{T}$
is denoted by $\phi(T,x)=\det(xI-A_{T})$.

A \emph{matching} is set of edges having pairwise distinct endpoints. A
$k$\emph{-matching} is a matching consisting of $k$ edges. A graph is called
\emph{vertex-transitive} if for every vertex pair $u$ and $v$, there exists an
automorphism $\varphi$ of the graph for which $\varphi(u)=v$. 

\subsection{Matching measure and tree of self-avoiding walks}

The \emph{matching polynomial} of a finite graph $G$ is defined as
\[
\mu(G,x)=\sum_{k}(-1)^{k}m_{k}(G)x^{\left\vert G\right\vert -2k},
\]
where $m_{k}(G)$ denotes the number of $k$-matchings in $G$. Let $\rho_{G}$,
the \emph{matching measure} of $G$ be the uniform distribution on the zeros of
the matching polynomial of $G$.

The fundamental theorem for the matching polynomial is the following.

\begin{Th}
[Heilmann and Lieb \cite{hei}]\label{Hei} The roots of the matching polynomial
$\mu(G,x)$ are real, and if the largest degree $D$ is greater than $1$, then
all roots lie in the interval $[-2\sqrt{D-1},2\sqrt{D-1}]$.
\end{Th}

A walk in a graph is \emph{self-avoiding} if it touches every vertex at most
once. For a finite graph $G$ and a root vertex $v$, one can construct
$T_{v}(G)$, the \emph{tree of self-avoiding walks at }$v$ as follows: its
vertices correspond to the finite self-avoiding walks in $G$ starting at $v$,
and we connect two walks if one of them is a one-step extension of the other.
The following figure illustrates that in general, $T_{v}(G)$ very much depends
on the choice of $v$.

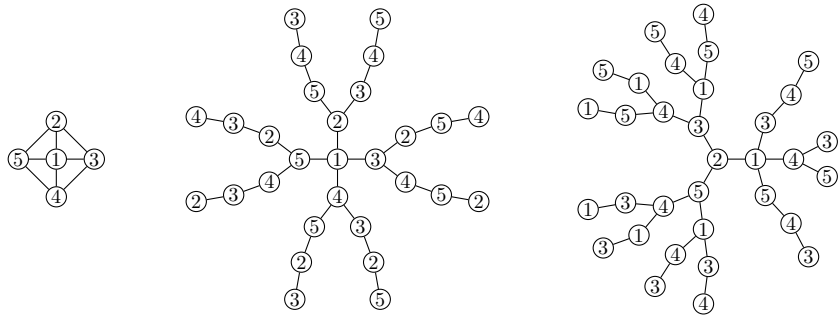
\begin{figure}[h!]
\centering
\tikzsetnextfilename{sawtree}
\begin{tikzpicture}
\begin{scope}[every node/.style={circle, draw, scale=.75, inner sep=.25mm}, xshift=-3.7cm]
\node (1) at (0,0) {1}; \node (2) at (0,.5) {2}; \node (3) at (.5,0) {3}; \node (4) at (0,-.5) {4}; \node (5) at (-.5,0) {5};
\draw (1) -- (2) (1) -- (3) (1) -- (4) (1) -- (5) (2) -- (3) (2) -- (5) (3) -- (4) (4) -- (5);
\end{scope}
\begin{scope}[grow cyclic, rotate=-45, level distance=5mm, inner sep=.25mm,
level 1/.style={sibling angle=-90}, level 2/.style={sibling angle=-75}, level 3/.style={sibling angle=-35},
level 4/.style={sibling angle=-20}, level 5/.style={sibling angle=-3}, level 6/.style={sibling angle=0},
every node/.style={circle, draw, thin, scale=.75}]
\node(1){1}child{node(12){2}child{node(125){5}child[missing]child{node(1254){4}child[missing]child{node(12543){3}
child[missing]child[missing]}}}child{node(123){3}child{node(1234){4}child{node(12345){5}child[missing]child[missing]}
child[missing]}child[missing]}}child{node(13){3}child{node(132){2}child[missing]child{node(1325){5}child[missing]
child{node(13254){4}child[missing]child[missing]}}}child{node(134){4}child{node(1345){5}child{node(13452){2}
child[missing]child[missing]}child[missing]}child[missing]}}child{node(14){4}child{node(143){3}child[missing]
child{node(1432){2}child[missing]child{node(14325){5}child[missing]child[missing]}}}child{node(145){5}
child{node(1452){2}child{node(14523){3}child[missing]child[missing]}child[missing]}child[missing]}}child{node(15){5}
child{node(154){4}child[missing]child{node(1543){3}child[missing]child{node(15432){2}child[missing]child[missing]}}}
child{node(152){2}child{node(1523){3}child{node(15234){4}child[missing]child[missing]}child[missing]}child[missing]}};
\end{scope}
\begin{scope}[grow cyclic, level distance=5mm, inner sep=.25mm, xshift=5cm,
level 1/.style={sibling angle=-120}, level 2/.style={sibling angle=-75}, level 3/.style={sibling angle=-55},
level 4/.style={sibling angle=-30}, level 5/.style={sibling angle=-15}, level 6/.style={sibling angle=0},
every node/.style={circle, draw, thin, scale=.75}]
\node(2){2}child{node(23){3}child{node(234){4}child{node(2345){5}child[missing]child{node(23451){1}child[missing]
child[missing]child[missing]}}child{node(2341){1}child{node(23415){5}child[missing]child[missing]}child[missing]
child[missing]}}child{node(231){1}child{node(2314){4}child[missing]child{node(23145){5}child[missing]child[missing]}}
child{node(2315){5}child{node(23154){4}child[missing]child[missing]}child[missing]}child[missing]}}child{node(21){1}
child{node(213){3}child[missing]child{node(2134){4}child{node(21345){5}child[missing]child[missing]}child[missing]}}
child{node(214){4}child{node(2143){3}child[missing]child[missing]}child{node(2145){5}child[missing]child[missing]}}
child{node(215){5}child{node(2154){4}child[missing]child{node(21543){3}child[missing]child[missing]}}child[missing]}}
child{node(25){5}child{node(251){1}child[missing]child{node(2513){3}child[missing]child{node(25134){4}child[missing]
child[missing]}}child{node(2514){4}child{node(25143){3}child[missing]child[missing]}child[missing]}}child{node(254){4}
child{node(2541){1}child[missing]child[missing]child{node(25413){3}child[missing]child[missing]}}child{node(2543){3}
child{node(25431){1}child[missing]child[missing]child[missing]}child[missing]}}};
\end{scope}
\end{tikzpicture}
\newcommand{\circled}[1]{\raisebox{-1.5pt}{\tikzsetnextfilename{circled#1}\tikz\node[circle, draw, scale=.75, inner sep=.25mm] {#1};}}
\caption{The pyramid graph and its trees of self-avoiding walks starting from \protect\circled1 and \protect\circled2 respectively.}
\label{fig:sawtree}
\end{figure}

Recall that the spectral measure of a (possibly infinite) rooted graph $(T,v)$
is defined as follows. Assume that $T$ has bounded degree. Then the adjacency
operator $A_{T}:l^{2}(T)\rightarrow l^{2}(T)$ is bounded and self-adjoint,
hence it admits a spectral measure $P_{T}(X)$ ($X\subseteq\mathbb{R}$ Borel).
This is a projection-valued measure on $\mathbb{R}$ such that for any
polynomial $F(x)$ we have
\begin{equation}
F(A)=\int F(x)dP_{x}\tag{Sp}%
\end{equation}
where $P_{x}=P((-\infty,x))$. We define $\delta_{(T,v)}$, the \emph{spectral
measure of }$T$\emph{ at }$v$ by
\[
\delta_{(T,v)}(X)=\left\langle P_{T}(X)\chi_{v},P_{T}(X)\chi_{v}\right\rangle
=\left\langle P_{T}(X)\chi_{v},\chi_{v}\right\rangle \text{ \ (}%
X\subseteq\mathbb{R}\text{ Borel)}%
\]
where $\chi_{v}$ is the characteristic vector of $v$. It is easy to check that
$\delta_{(T,v)}$ is a probability measure supported on the spectrum of the
operator $A_{T}$. Also, by (Sp), for all $k\geq0$, the $k$-th moment of
$\delta_{(T,v)}$ equals
\[
\int x^{k}d\delta_{(T,v)}=\left\langle A^{k}\chi_{v},\chi_{v}\right\rangle
=a_{k}(T,v)
\]
where $a_{k}(T,v)$ is the number of returning walks of length $k$ starting at
$v$.

It turns out that the matching measure of a finite graph equals the average
spectral measure over its trees of self-avoiding walks.

\begin{Th}
\label{expected} Let $G$ be a finite graph and let $v$ be a vertex of $G$
chosen uniformly at random. Then
\[
\rho_{G}=\mathbb{E}_{v}\delta_{(T_{v}(G),v)}\text{.}%
\]
Equivalently, for all $k\geq0$, the $k$-th moment of $\rho_{G}$ equals the
expected number of returning walks of length $k$ in $T_{v}(G)$ starting at $v$.
\end{Th}

In particular, Theorem \ref{expected} gives one of the several known proofs
for the Heilmann-Lieb theorem. Indeed, spectral measures are real and the
spectral radius of a tree with degree bound $D$ is at most $2\sqrt{D-1}$.

To prove Theorem~\ref{expected} we need the following result of Godsil
\cite{god3} which connects the matching polynomial of the original graph $G$
and the tree of self-avoiding walks:

\begin{Th}
\label{saw-tree} \cite{god3} Let $G$ be a finite graph and $v$ be an arbitrary
vertex of $G$. Then
\[
\frac{\mu(G-v,x)}{\mu(G,x)}=\frac{\mu(T_{v}(G)-v,x)}{\mu(T_{v}(G),x)}.
\]

\end{Th}

We will also use two well-known facts which we gather in the following proposition:

\begin{Prop}
\label{prop} \cite{god3} (a) For any tree or forest $T$, the matching
polynomial $\mu(T,x)$ coincides with the characteristic polynomial $\phi(T,x)$
of the adjacency matrix of the tree $T$:
\[
\mu(T,x)=\phi(T,x).
\]
\medskip

\noindent(b) For any graph $G$, we have
\[
\mu^{\prime}(G,x)=\sum_{v\in V}\mu(G-v,x).
\]

\end{Prop}

\begin{proof}
[Proof of Theorem~\ref{expected}]First, let us use part (a) of
Proposition~\ref{prop} for the tree $T_{v}(G)$ and the forest $T_{v}(G)-v$:
\[
\frac{\mu(T_{v}(G)-v,x)}{\mu(T_{v}(G),x)}=\frac{\phi(T_{v}(G)-v,x)}{\phi
(T_{v}(G),x)}.
\]
On the other hand, for any graph $H$ and vertex $u$, we have
\[
\frac{\phi(H-u,x)}{\phi(H,x)}=x^{-1}\sum_{k=0}^{\infty}c_{k}(u)x^{-k},
\]
where $c_{k}(u)$ counts the number of walks of length $k$ starting and ending
at $u$. So this is exactly the moment generating function of the spectral
measure with respect to the vertex $u$. Putting together these with
Theorem~\ref{saw-tree} we see that
\[
\frac{\mu(G-v,x)}{\mu(G,x)}=\frac{\mu(T_{v}(G)-v,x)}{\mu(T_{v}(G),x)}%
=x^{-1}\sum_{k=0}^{\infty}a_{k}(v)x^{-k}%
\]
is the moment generating function of the spectral measure of the tree of
self-avoiding walks with respect to the vertex $v$.

Now let us consider the left hand side of Theorem~\ref{saw-tree}. Let us use
part (b) of Proposition~\ref{prop}:
\[
\mu^{\prime}(G,x)=\sum_{u\in V}\mu(G-u,x).
\]
This implies that
\[
\mathbb{E}_{v}\frac{\mu(G-v,x)}{\mu(G,x)}=\frac{1}{|G|}\frac{\mu^{\prime
}(G,x)}{\mu(G,x)}=x^{-1}\sum_{k=0}^{\infty}\mu_{k}x^{-k},
\]
where
\[
\mu_{k}=\frac{1}{|G|}\sum\lambda^{k},
\]
where the summation goes through the zeros of the matching polynomial. In
other words, $\mu_{k}$ is $k$-th moment of the matching measure defined by the
uniform distribution on the zeros of the matching polynomial. Putting
everything together we see that
\[
\mu_{k}=\mathbb{E}_{v}a_{k}(v).
\]
Since both $\rho_{G}$ and $\mathbb{E}_{v}\rho(v)$ are supported on $\left\{
\left\vert x\right\vert \leq\left\Vert A_{G}\right\Vert \right\}  $, we get
that the two measures are equal.
\end{proof}

Now we define Benjamini--Schramm convergence. 

\begin{Def}
For a finite graph $G$, a finite rooted graph $\alpha$ and a positive integer
$r$, let $\mathbb{P}(G,\alpha,r)$ be the probability that the $r$-ball
centered at a uniform random vertex of $G$ is isomorphic to $\alpha$. We say
that a graph sequence $(G_{n})$ of bounded degree is \emph{Benjamini--Schramm
convergent} if for all finite rooted graphs $\alpha$ and $r>0$, the
probabilities $\mathbb{P}(G_{n},\alpha,r)$ converge. Let $L$ be a vertex
transitive lattice. We say that {\emph{$(G_{n})$ Benjamini-Schramm converges to $L$}},
if for all positive integers $r$, $\mathbb{P}(G_{n},\alpha_{r},r)\rightarrow1$
where $\alpha_{r}$ is the $r$-ball in $L$. 
\end{Def}

\begin{Ex}
Let us consider a sequence of boxes in $\mathbb{Z}^{d}$ where all sides
converge to infinity. This will be Benjamini--Schramm convergent graph
sequence since for every fixed $r$, we will pick a vertex which at least
$r$-far from the boundary with probability converging to $1$. For all these
vertices we will see the same neighborhood. This also shows that we can impose
arbitrary boundary condition, for instance periodic boundary condition means
that we consider the sequence of toroidal boxes. Boxes and toroidal boxes will
be Benjamini--Schramm convergent even together.
\end{Ex}

We prove the following generalization of Theorem~\ref{konv}. 

\begin{Th}
\label{wc} Let $(G_{n})$ be a Benjamini--Schramm convergent bounded degree
graph sequence. Then the sequence of matching measures $\rho_{G_{n}}$ is
weakly convergent. If $(G_{n})$ Benjamini--Schramm converges to the vertex
transitive lattice $L$, then $\rho_{G_{n}}$ weakly converges to $\rho_{L}$ and
$\lim_{n\rightarrow\infty}\rho_{G_{n}}(\{x\})=\rho_{L}(\{x\})$ for all
$x\in\mathbb{R}$. 
\end{Th}

\begin{Rem}
The first part of the theorem was first proved in \cite{ACFK}. The proof given
there relied on a general result on graph polynomials given in \cite{csi}. For completeness, we give an
alternate self-contained proof here. 
\end{Rem}

We will use the following theorem of Thom \cite{thom}. See also
\cite{luckapprox} where this is used for Benjamini--Schramm convergent graph
sequences. 

\begin{Th}[Thom]\label{Luck} Let $(q_{n}(z))$ be a sequence of monic polynomials with
integer coefficients. Assume that all zeros of all $q_{n}(z)$ are at most $R$
in absolute value. Let $\rho_{n}$ be the probability measure of uniform
distribution on the roots of $q_{n}(z)$. Assume that $\rho_{n}$ weakly
converges to some measure $\rho$. Then for all $\theta\in\mathbb{C}$ we have
\[
\lim_{n\rightarrow\infty}\rho_{n}(\{\theta\})=\rho(\{\theta\}).
\]

\end{Th}

\begin{proof}[Proof of Theorem~\ref{konv} and \ref{wc}]
For $k\geq0$ let
\[
\mu_{k}(G)=\int z^{k}\,d\rho_{G}(z)
\]
be the $k$-th moment of $\rho_{G}$. By Theorem~\ref{expected} we have
\[
\mu_{k}(G)=\mathbb{E}_{v}a_{k}(G,v)
\]
where $a_{k}(G,v)$ denotes the number of closed walks of length $k$ of the
tree $T_{v}(G)$ starting and ending at the vertex $v$. 

Clearly, the value of $a_{k}(G,v)$ only depends on the $k$-ball centered at
the vertex $v$. Let $TW(\alpha)=a_{k}(G,v)$ where the $k$-ball centered at $v$
is isomorphic to $\alpha$. Note that the value of $TW(\alpha)$ depends only on the rooted graph $\alpha$ and does not depend on $G$.

Let $\mathcal{N}_{k}$ denote the set of possible $k$-balls in $G$. The size of
$\mathcal{N}_{k}$ and $TW(\alpha)$ are bounded by a function of $k$ and the largest
degree of $G$. By the above, we have
\[
\mu_{k}(G)=\mathbb{E}_{v}a_{k}(G,v)=\sum_{\alpha\in\mathcal{N}_{k}}%
\mathbb{P}(G,\alpha,k)\cdot TW(\alpha).
\]
Since $(G_{n})$ is Benjamini--Schramm convergent, we get that for every fixed
$k$, the sequence of $k$-th moments $\mu_{k}(G_{n})$ converges. The same holds
for $\int q(z)\,d\rho_{G_{n}}(z)$ where $q$ is any polynomial. By the
Heilmann--Lieb theorem, $\rho_{G_{n}}$ is supported on $[-2\sqrt{D-1}%
,2\sqrt{D-1}]$ where $D$ is the absolute degree bound for $G_{n}$. Since every
continuous function can be uniformly approximated by a polynomial on
$[-2\sqrt{D-1},2\sqrt{D-1}]$, we get that the sequence $(\rho_{G_{n}})$ is
weakly convergent. 

Assume that $(G_{n})$ Benjamini--Schramm converges to $L$. Then for all
$k\geq0$ we have $\mathbb{P}(G_{n},\alpha_{k},k)\rightarrow 1$ where
$\alpha_{k}$ is the $k$-ball in $L$, which implies
\[
\lim_{n\rightarrow\infty}\mu_{k}(G_{n})=\lim_{n\rightarrow\infty}\sum
_{\alpha\in\mathcal{N}_{k}}\mathbb{P}(G_n,\alpha,k)\cdot TW(\alpha
)=TW(\alpha_{k})=a_{k}(L,v)
\]
where $v$ is any vertex in $L$. This means that all the moments of $\rho_{L}$
and $\lim\rho_{G_{n}}$ are equal, so $\lim\rho_{G_{n}}=\rho_{L}$. 

Since the matching polynomial is monic with integer coefficients, Theorem
\ref{Luck} gives $\lim_{n\rightarrow\infty}\rho_{G_{n}}(\{x\})=\rho
_{L}(\{x\})$ for all $x\in\mathbb{R}$. 
\end{proof}

\section{The function $\lambda_{G}(p)$}

\label{entropy-function}

Let $G$ be a finite graph, and recall that $|G|$ denotes the number of vertices of
$G$, and $m_{k}(G)$ denotes the number of $k$-matchings ($m_{0}%
(G)=1$). Let $t$ be the activity, a non-negative real number, and
\[
M(G,t)=\sum_{k=0}^{\lfloor|G|/2\rfloor}m_{k}(G)t^{k},
\]
We call $M(G,t)$ the matching generating function or the partition function of
the monomer-dimer model. Clearly, it encodes the same information as the
matching polynomial. Let
\[
p(G,t)=\frac{2t\cdot M^{\prime}(G,t)}{|G|\cdot M(G,t)},
\]
and
\[
F(G,t)=\frac{\ln M(G,t)}{|G|}-\frac{1}{2}p(G,t) \ln(t).
\]
Note that
\[
\tilde{\lambda}(G)=F(G,1)
\]
is called the monomer-dimer free energy.

The function $p=p(G,t)$ is a strictly monotone increasing function which maps
$[0,\infty)$ to $[0,p^{*})$, where $p^{*}=\frac{2\nu(G)}{|G|}$, where $\nu(G)
$ denotes the number of edges in the largest matching. If $G$ contains a
perfect matching, then $p^{*}=1$. Therefore, its inverse function $t=t(G,p)$
maps $[0,p^{*})$ to $[0,\infty)$. (If $G$ is clear from the context, then we will simply write $t(p)$ instead of $t(G,p)$.) Let
\[
\lambda_{G}(p)=F(G,t(p))
\]
if $p<p^{*}$, and $\lambda_{G}(p)=0$ if $p>p^{*}$. Note that we have not
defined $\lambda_{G}(p^{*})$ yet. We simply define it as a limit:
\[
\lambda_{G}(p^{*})=\lim_{p\nearrow p^{*}}\lambda_{G}(p).
\]
We will show that this limit exists, see part (d) of Proposition~\ref{asymp}.
Later we will extend the definition of $p(G,t), F(G,t)$ and $\lambda_{G}(p) $
to infinite lattices $L$.

The intuitive meaning of $\lambda_{G}(p)$ is the following. Assume that we
want to count the number of matchings covering $p$ fraction of the vertices.
Let us assume that it makes sense: $p=\frac{2k}{|G|}$, and so we wish to count
$m_{k}(G)$. Then
\[
\lambda_{G}(p)\approx\frac{\ln m_{k}(G)}{|G|}.
\]
The more precise formulation of this statement will be given in
Proposition~\ref{asymp}. To prove this proposition we need some preparation.
\bigskip

We will use the following theorem of Darroch.

\begin{Lemma}
[Darroch's rule \cite{dar}]Let $P(x)=\sum_{k=0}^{n}a_{k}x^{k}$ be a
polynomial with only positive coefficients and real zeros. If
\[
k-\frac{1}{n-k+2}<\frac{P^{\prime}(1)}{P(1)}<k+\frac{1}{k+2},
\]
then $k$ is the unique number for which $a_{k}=\max(a_{1},a_{2},\dots, a_{n})
$. If, on the other hand,
\[
k+\frac{1}{k+2}<\frac{P^{\prime}(1)}{P(1)}<k+1-\frac{1}{n-k+1},
\]
then either $a_{k}$ or $a_{k+1}$ is the maximal element of $a_{1},a_{2},\dots,
a_{n}$.
\end{Lemma}

\begin{Prop}
\label{asymp} Let $G$ be a finite graph. \newline(a) Let $nG$ be $n$ disjoint
copies of $G$. Then
\[
\lambda_{G}(p)=\lambda_{nG}(p).
\]
(b) If $p<p^{*}$, then
\[
\frac{d}{dp}\lambda_{G}(p)=-\frac{1}{2}\ln t(p).
\]
(c) The limit
\[
\lambda_{G}(p^*)=\lim_{p\nearrow p^{*}}\lambda_{G}(p)
\]
exists. 
\newline(d) Let $k\leq\nu(G)$ and $p=\frac{2k}{|G|}$. Then
\[
\left|\lambda_{G}(p)-\frac{\ln m_{k}(G)}{|G|}\right|  \leq\frac{\ln|G|}%
{|G|}.
\]
(e) Let $k=\nu(G)$, then for $p^{*}=\frac{2k}{|G|}$ we have
\[
\lambda_{G}(p^{*})=\frac{\ln m_{k}(G)}{|G|}.
\]
(f) If for some function $f(p)$ we have
\[
\lambda_{G}(p)\geq f(p)+o_{|G|}(1)
\]
then
\[
\lambda_{G}(p)\geq f(p).
\]

\end{Prop}

\begin{proof}
\noindent(a) Let $nG$ be the disjoint union of $n$ copies of $G$. Note that
\[
M(nG,t)=M(G,t)^{n}%
\]
implying that $p(nG,t)=p(G,t)$ and $\lambda_{nG}(p)=\lambda_{G}(p)$. \medskip

\noindent(b) Since
\[
\lambda_{G}(p)=\frac{\ln M(G,t)}{|G|}-\frac{1}{2}p(G,t) \ln(t)
\]
we have
\[
\frac{d\lambda_{G}(p)}{dp}=\left(  \frac{1}{|G|}\cdot\frac{M^{\prime}%
(G,t)}{M(G,t)}\cdot\frac{dt}{dp}-\frac{1}{2}\left(  \ln(t)+p\cdot\frac{1}%
{t}\cdot\frac{dt}{dp}\right)  \right)  =-\frac{1}{2}\ln(t),
\]
since
\[
\frac{1}{|G|}\cdot\frac{M^{\prime}(G,t)}{M(G,t)}=\frac{p}{2t}%
\]
by definition. \medskip

\noindent(c) From $\frac{d}{dp}\lambda_{G}(p)=-\frac{1}{2}\ln t(p)$ we see that if $p>p(G,1)$, the function $\lambda_G(p)$ is
monotone decreasing. (Note that we also see
that $\lambda_{G}(p)$ is a concave-down function.) Hence
\[
\lim_{p\nearrow p^{*}}\lambda_{G}(p)=\inf_{p>p(G,1)}\lambda_{G}(p).
\]
\medskip

\noindent(d) First, let us assume that $k<\nu(G)$. In case of $k=\nu(G)$, we
will slightly modify our argument. Let $t=t(p)$ be the value for which
$p=p(G,t)$. The polynomial
\[
P(G,x)=M(G,tx)=\sum_{j=0}^{n}m_{j}(G)t^{j}x^{j}%
\]
considered as a polynomial in variable $x$, has only real zeros by
Theorem~\ref{Hei}. Note that
\[
k=\frac{p|G|}{2}=\frac{P^{\prime}(G,1)}{P(G,1)}.
\]
Darroch's rule says that in this case $m_{k}(G)t^{k}$ is the unique maximal
element of the coefficient sequence of $P(G,x)$. In particular
\[
\frac{M(G,t)}{|G|}\leq m_{k}(G)t^{k}\leq M(G,t).
\]
Hence
\[
\lambda_{G}(p)-\frac{\ln|G|}{|G|}\leq\frac{\ln m_{k}(G)}{|G|}\leq\lambda
_{G}(p).
\]
Hence in case of $k<\nu(G)$, we are done. \medskip

If $k=\nu(G)$, then let $p$ be arbitrary such that
\[
k-\frac{1}{2}<\frac{p|G|}{2}<k.
\]
Again we can argue by Darroch's rule as before that
\[
\lambda_{G}(p)-\frac{\ln|G|}{|G|}\leq\frac{\ln m_{k}(G)}{|G|}\leq\lambda
_{G}(p).
\]
Since this is true for all $p$ sufficiently close to $p^{*}=\frac{2\nu
(G)}{|G|}$ and
\[
\lambda_{G}(p^{*})=\lim_{p\nearrow p^{*}}\lambda_{G}(p),
\]
we have
\[
\left|  \frac{\ln m_{k}(G)}{|G|}-\lambda_{G}(p^{*})\right|  \leq\frac{\ln
|G|}{|G|}%
\]
in this case too. \medskip

\noindent(e) By part (a) we have $\lambda_{nG}(p)=\lambda_{G}(p)$. Note also
that if $k=\nu(G)$, then $m_{nk}(nG)=m_{k}(G)^{n}$. Applying the bound from
part (d) to the graph $nG$, we obtain that
\[
\left|  \frac{\ln m_{k}(G)}{|G|}-\lambda_{G}(p^{*})\right|  \leq\frac{\ln
|nG|}{|nG|}.
\]
Since
\[
\frac{\ln|nG|}{|nG|}\to0
\]
as $n\to\infty$, we get that
\[
\lambda_{G}(p^{*})=\frac{\ln m_{k}(G)}{|G|}.
\]
\medskip

\noindent(f) This is again a trivial consequence of $\lambda_{nG}%
(p)=\lambda_{G}(p)$.
\end{proof}

Our next aim is to extend the definition of the function $\lambda_{G}(p)$ for
infinite lattices $L$. We also show an efficient way of computing its values
if $p$ is sufficiently separated from $p^{*}$.

The following theorem was known in many cases for thermodynamic limit.

\begin{Th}
\label{entropy} Let $(G_{n})$ be a Benjamini--Schramm convergent 
sequence of bounded degree graphs. Then the sequences of functions\newline(a)
\[
p(G_{n},t),
\]
(b)
\[
\frac{\ln M(G_{n},t)}{|G_{n}|}%
\]
converge to strictly monotone increasing continuous functions on the interval
$[0,\infty)$. \newline If, in addition, every $G_{n}$ has a perfect matching
then the sequences of functions \newline(c)
\[
t(G_{n},p),
\]
(d)
\[
\lambda_{G_{n}}(p)
\]
are convergent for all $0\leq p<1$.
\end{Th}

\begin{Rem}
In part (c), we used the extra condition to ensure that $p^{*}=1$ for all
these graphs. We mention that H. Nguyen and K. Onak \cite{ngu}, and
independently G. Elek and G. Lippner \cite{ele} proved that for a
Benjamini--Schramm convergent graph sequence $(G_{n})$, the following limit
exits:
\[
\lim_{n\to\infty}\frac{2\nu(G_{n})}{|G_{n}|}=\lim_{n\to\infty}p^{*}(G_{n}).
\]
In particular, one can extend part (c) to graph sequences without perfect
matchings. Since we are primarily interested in lattices with perfect matchings,
we leave it to the Reader.
\end{Rem}

To prove Theorem~\ref{entropy}, we essentially repeat an argument of the paper
\cite{ACFK}.

\begin{proof}
[Proof of Theorem~\ref{entropy}]First we prove part (a) and (b). For a graph
$G$ let $S(G)$ denote the set of zeros of the matching polynomial $\mu(G,x)$,
then
\[
M(G,t)=\prod_{%
\genfrac{}{}{0pt}{}{\lambda\in S(G) }{\lambda>0}%
}(1+\lambda^{2}t)=\prod_{\lambda\in S(G)}(1+\lambda^{2}t)^{1/2}.
\]
Then
\[
\ln M(G,t)=\sum_{\lambda\in S(G)}\frac{1}{2}\ln\left(  1+\lambda^{2}t\right)
.
\]
By differentiating both sides we get that
\[
\frac{M^{\prime}(G,t)}{M(G,t)}=\sum_{\lambda\in S(G)}\frac{1}{2}\frac
{\lambda^{2}}{1+\lambda^{2}t}.
\]
Hence
\[
p(G,t)=\frac{2t\cdot M^{\prime}(G,t)}{|G|\cdot M(G,t)}=\frac{1}{|G|}%
\sum_{\lambda\in S(G)}\frac{\lambda^{2} t}{1+\lambda^{2}t}=\int\frac{tz^{2}%
}{1+tz^{2}}\, d\rho_{G}(z).
\]
Similarly,
\[
\frac{\ln M(G,t)}{|G|}=\frac{1}{|G|}\sum_{\lambda\in S(G)}\frac{1}{2}%
\ln\left(  1+\lambda^{2}t\right)  =\int\frac{1}{2}\ln\left(  1+tz^{2}\right)
\, d\rho_{G}(z).
\]
Since $(G_{n})$ is a Benjamini--Schramm convergent sequence of bounded degree graphs, 
the sequence $(\rho_{G_{n}})$  weakly converges to some $\rho^{*}
$ by Theorem~\ref{wc}. Since both functions
\[
\frac{tz^{2}}{1+tz^{2}}\ \ \ \ \mbox{and}\ \ \ \ \frac{1}{2}\ln\left(
1+tz^{2}\right)
\]
are continuous, we immediately obtain that
\[
\lim_{n\to\infty}p(G_{n},t)=\int\frac{tz^{2}}{1+tz^{2}} \, d\rho^{*}(z),
\]
and
\[
\lim_{n\to\infty}\frac{\ln M(G_{n},t)}{|G_{n}|}=\int\frac{1}{2}\ln\left(
1+tz^{2}\right)  \, d\rho^{*}(z).
\]
Note that both functions
\[
\frac{tz^{2}}{1+tz^{2}}\ \ \ \ \mbox{and}\ \ \ \ \frac{1}{2}\ln\left(
1+tz^{2}\right)
\]
are strictly monotone increasing continuous functions in the variable $t$.
Thus their integrals are also strictly monotone increasing continuous
functions. \bigskip

To prove part (c), let us introduce the function
\[
p(L,t)=\int\frac{tz^{2}}{1+tz^{2}} \, d\rho^{*}(z).
\]
We have seen that $p(L,t)$ is a strictly monotone increasing continuous
function, and equals $\lim_{n\to\infty}p(G_{n},t)$. Since for all $G_{n}$,
$p^{*}(G_{n})=1$, we have $\lim_{t \to\infty}p(G_{n},t)=1$ for all $n$. This
means that $\lim_{t \to\infty}p(L,t)=1$. Hence we can consider inverse
function $t(L,p)$  which maps $[0,1)$ to $[0,\infty)$. We show that
\[
\lim_{n\to\infty}t(G_{n},p)=t(L,p)
\]
pointwise. Assume by contradiction that this is not the case. This means that
for some $p_{1}$, there exists an $\varepsilon$ and an infinite sequence
$n_{i}$ for which
\[
\left|  t(L,p_{1})-t(G_{n_{i}},p_{1})\right|  \geq\varepsilon.
\]
We distinguish two cases according to \newline(i) there exists an infinite
sequence $(n_{i})$ for which
\[
t(G_{n_{i}},p_{1})\geq t(L,p_{1})+\varepsilon,
\]
or (ii) there exists an infinite sequence $(n_{i})$ for which
\[
t(G_{n_{i}},p_{1})\leq t(L,p_{1})-\varepsilon.
\]
In the first case, let $t_{1}=t(L,p_{1})$, $t_{2}=t_{1}+\varepsilon$ and
$p_{2}=p(L,t_{2})$. Clearly, $p_{2}>p_{1}$. Note that
\[
t(G_{n_{i}},p_{1})\geq t(L,p_{1})+\varepsilon=t_{2}%
\]
and $p(G_{n_{i}},t)$ are monotone increasing functions, thus
\[
p(G_{n_{i}},t_{2})\leq p(G_{n_{i}},t(G_{n_{i}},p_{1}))=p_{1}=p_{2}%
-(p_{2}-p_{1})=p(L,t_{2})-(p_{2}-p_{1}).
\]
This contradicts the fact that
\[
\lim_{n\to\infty}p(G_{n_{i}},t_{2})=p(L,t_{2}).
\]
In the second case, let $t_{1}=t(L,p_{1})$, $t_{2}=t_{1}-\varepsilon$ and
$p_{2}=p(L,t_{2})$. Clearly, $p_{2}<p_{1}$. Note that
\[
t(G_{n_{i}},p_{1})\leq t(L,p_{1})-\varepsilon=t_{2}%
\]
and $p(G_{n_{i}},t)$ are monotone increasing functions, thus
\[
p(G_{n_{i}},t_{2})\geq p(G_{n_{i}},t(G_{n_{i}},p_{1}))=p_{1}=p_{2}%
+(p_{1}-p_{2})=p(L,t_{2})+(p_{1}-p_{2}).
\]
This again contradicts the fact that
\[
\lim_{n\to\infty}p(G_{n_{i}},t_{2})=p(L,t_{2}).
\]
Hence $\lim_{n\to\infty}t(G_{n},p)=t(L,p)$. \medskip

Finally, we show that $\lambda_{G_{n}}(p)$ converges for all $p$. Let
$t=t(L,p) $, and
\[
\lambda_{L}(p)=\lim_{n\to\infty}\frac{\ln M(G_{n},t)}{|G_{n}|}-\frac{1}{2}%
p\ln(t).
\]
Note that
\[
\lambda_{G_{n}}(p)=\frac{\ln M(G_{n},t_{n})}{|G_{n}|}-\frac{1}{2}p\ln(t_{n}),
\]
where $t_{n}=t(G_{n},p)$. We have seen that $\lim_{n\to\infty}t_{n}=t$. Hence
it is enough to prove that the functions
\[
\frac{\ln M(G_{n},u)}{|G_{n}|}%
\]
are equicontinuous. Let us fix some $u_{0}$ and let
\[
H(u_{0},u)=\max_{z\in[-2\sqrt{D-1},2\sqrt{D-1}]}\left|  \frac{1}{2}\ln\left(
1+u_{0}z^{2}\right)  -\frac{1}{2}\ln\left(  1+uz^{2}\right)  \right|  .
\]
Clearly, if $|u-u_{0}|\leq\delta$ for some sufficiently small $\delta$, then
$H(u_{0},u)\leq\varepsilon$, and
\[
\left|  \frac{\ln M(G_{n},u)}{|G_{n}|}-\frac{\ln M(G_{n},u_{0})}{v(G_{n}%
)}\right|  =\left|  \int\frac{1}{2}\ln\left(  1+u_{0}z^{2}\right)  \,
d\rho_{G_{n}}(z)-\int\frac{1}{2}\ln\left(  1+uz^{2}\right)  \, d\rho_{G_{n}%
}(z)\right|  \leq
\]
\[
\leq\int\left|  \frac{1}{2}\ln\left(  1+u_{0}z^{2}\right)  -\frac{1}{2}%
\ln\left(  1+uz^{2}\right)  \right|  \, d\rho_{G_{n}}(z)\leq\int H(u,u_{0})\,
d\rho_{G_{n}}(z)\leq\varepsilon.
\]
This completes the proof of the convergence of $\lambda_{G_{n}}(p)$.
\end{proof}

\begin{Def}
Let $L$ be an infinite lattice and $(G_{n})$ be a sequence of finite graphs
which is Benjamini--Schramm convergent to $L$. For instance, $G_{n}$ can be
chosen to be an exhaustion of $L$. Then the sequence of measures $(\rho
_{G_{n}})$ weakly converges to some measure which we will call $\rho_{L}$, the
matching measure of the lattice $L$. For $t>0$, we can introduce
\[
p(L,t)=\int\frac{tz^{2}}{1+tz^{2}}\, d\rho_{L}(z)
\]
and
\[
F(L,t)=\int\frac{1}{2}\ln\left(  1+tz^{2}\right)  \, d\rho_{L}(z)-\frac{1}%
{2}p(L,t) \ln(t).
\]
If the lattice $L$ contains a perfect matching, then we can choose $G_{n}$
such that all $G_{n}$ contain a perfect matching. Then $p(L,t)$ maps
$[0,\infty)$ to $[0,1)$ in a monotone increasing way, and we can consider its
inverse function $t(L,p)$. Finally, we can introduce
\[
\lambda_{L}(p)=F(L,t(L,p))
\]
for all $p\in[0,1)$. We will define $\lambda_{L}(1)$ as
\[
\lambda_{L}(1)=\lim_{p\nearrow1}\lambda_{L}(p).
\]

\end{Def}

\begin{Rem}
\label{Mayer} In the literature, the so-called Mayer series are computed for
various lattices $L$:
\[
p(L,t)=\sum_{n=1}^{\infty}a_{n}t^{n}%
\]
for small enough $t$. Let us compare it with
\[
p(L,t)=\int\frac{tz^{2}}{1+tz^{2}}\, d\rho_{L}(z)=\int\left(  \sum
_{n=1}^{\infty}(-1)^{n+1}z^{2n}t^{n}\right)  \, d\rho_{L}(z)=\sum
_{n=1}^{\infty}(-1)^{n+1}\left(  \int z^{2n}d\rho_{L}(z)\right)  t^{n}.
\]
Hence if we introduce the moment sequence
\[
\mu_{k}=\int z^{k}d\rho_{L}(z),
\]
we see that
\[
\mu_{2n}=\int z^{2n}d\rho_{L}(z)=(-1)^{n+1}a_{n}.
\]
Note that $\mu_{0}=1$ and $\mu_{2n-1}=0$ since the matching measures are
symmetric to $0$. Since the support of the measure $\rho_{L}$ lie in the
interval $[-2\sqrt{D-1},2\sqrt{D-1}]$, we see that the Mayer series converges
whenever $|t|<\frac{1}{4(D-1)}$. We also would like to point out that the
integral is valid for all $t>0$, while the Mayer series does not converge if
$t$ is 'large'.
\end{Rem}

\subsection{Computation of the monomer-dimer free energy}

The monomer-dimer free energy of a lattice $L$ is $\tilde{\lambda}(L)=F(L,1)$.
Its computation can be carried out exactly the same way as we proved its
existence: we use that
\[
\tilde{\lambda}(L)=F(L,1)=\int\frac{1}{2}\ln\left(  1+z^{2}\right)  \,d\rho_{L}(z).
\]
Assume that we know the moment sequence $(\mu_{k})$ for $k\leq N$. Then let us
choose a polynomial of degree at most $N$, which uniformly approximates the
function
\[
\frac{1}{2}\ln\left(  1+z^{2}\right)
\]
on the interval $[-2\sqrt{D-1},2\sqrt{D-1}]$, where $D$ is the coordination
number of $L$. A good polynomial approximation can be found by Remez's
algorithm. Assume that we have a polynomial
\[
q(z)=\sum_{k=0}^{N}c_{k}z^{k}%
\]
for which
\[
\left\vert \frac{1}{2}\ln\left(  1+z^{2}\right)  -q(z)\right\vert
\leq\varepsilon
\]
for all $z\in\lbrack-2\sqrt{D-1},2\sqrt{D-1}]$. Then
\[
\left\vert \tilde{\lambda}(L)-\int q(z)\,d\rho_{L}(z)\right\vert \leq\int\left\vert
\frac{1}{2}\ln\left(  1+z^{2}\right)  -q(z)\right\vert d\rho_{L}%
(z)\leq\varepsilon,
\]
and
\[
\int q(z)\,d\rho_{L}(z)=\sum_{k=0}^{N}c_{k}\mu_{k}.
\]
Hence
\[
\left\vert \tilde{\lambda}(L)-\sum_{k=0}^{N}c_{k}\mu_{k}\right\vert \leq\varepsilon.
\]
How can we compute the moment sequence $(\mu_{k})$? One way is to use its
connection with the Mayer series (see Remark~\ref{Mayer}). A good source of Mayer
series coefficients is the paper of P. Butera and M. Pernici \cite{bupe},
where they computed $a_{n}$ for $1\leq n\leq24$ for various lattices. (More
precisely, they computed $d_{n}=a_{n}/2$ since they expanded the function
$\rho(t)=p(t)/2$.) This means that we know  $\mu_{k}$ for $k\leq49$ for
these lattices. The other strategy to compute the moment sequence is to use
its connection with the number of closed walks in the self-avoiding walk tree. 

Since the moment sequence is missing for the honeycomb lattice (hexagonal lattice), we
computed the first few elements of the moment sequence for this lattice:
\medskip%

\[
1, 0, 3, 0, 15, 0, 87, 0, 543, 0, 3543, 0, 23817, 0, 163551, 0, 1141119, 0,
8060343, 0,
\]
\[
57494385, 0, 413383875, 0, 2991896721, 0, 21774730539, 0, 159227948055, 0,
\]
\[
1169137211487, 0, 8615182401087, 0, 63683991513351, 0, 472072258519041, 0,
\]
\[
3508080146139867, 0, 26127841824131313, 0, 194991952493587371, 0,
\]
\[
1457901080870060919, 0, 10918612274039599755, 0, 81898043907874542705
\]
\bigskip

The following table contains some numerical results. The bound on the error terms are rigorous.

\begin{center}%
\begin{tabular}[c]{|c|c|c|c|c|}\hline
Lattice & $\tilde{\lambda}(L)$ & Bound on error & $p(L,1)$ & Bound on error  \\\hline
2d & $0.6627989725$ & $3.72\cdot10^{-8}$ & $0.638123105$ &
$5.34\cdot10^{-7}$  \\\hline
3d & $0.7859659243$ & $9.89\cdot10^{-7}$ & $0.684380278$ &
$1.14\cdot10^{-5}$  \\\hline
4d & $0.8807178880$ & $5.92\cdot10^{-6}$ & $0.715846906$ &
$5.86\cdot10^{-5}$  \\\hline
5d & $0.9581235802$ & $4.02\cdot10^{-5}$ & $0.739160383$ &
$3.29\cdot10^{-4}$  \\\hline
6d & $1.0237319240$ & $1.24\cdot10^{-4}$ & $0.757362382$ &
$8.91\cdot10^{-4}$  \\\hline
7d & $1.0807591953$ & $3.04\cdot10^{-4}$ & $0.772099489$ &
$1.95\cdot10^{-3}$  \\\hline
hex &$0.58170036638$ & $1.56\cdot10^{-9}$ &$0.600508638$ &
$2.65\cdot10^{-8}$  \\\hline
\end{tabular}

\end{center}

\section{Density function of matching measures.}

\label{density}

It is natural problem to investigate the matching measure. One particular
question is whether it is atomless or not. In general, $\rho_{L}$ can contain
atoms. For instance, if $G$ is a finite graph then clearly $\rho_{G}$ consists
of atoms. On the other hand, it can be shown that for all lattices in Table 1,
the measure $\rho_{L}$ is atomless. We use the following lemmas. 

\begin{figure}[h!]
\tikzsetnextfilename{2dgridmm} \begin{tikzpicture}
\begin{axis}[small, height=4cm, width=12cm, xmin=-5, xmax=5, ymin=0, ymax=0.25, axis lines=left, /pgf/number format/fixed]
\addplot[black!80, fill=black, fill opacity=0.15] table {2dgridmm.dat};
\end{axis}
\end{tikzpicture}
\caption{An approximation for the matching measure of $\mathbb{Z}^{2}$,
obtained by smoothing the matching measure of the finite grid $C_{10}\times
P_{100}$ by convolution with a triweight kernel.}%
\label{fig:2dgridmm}%
\end{figure}
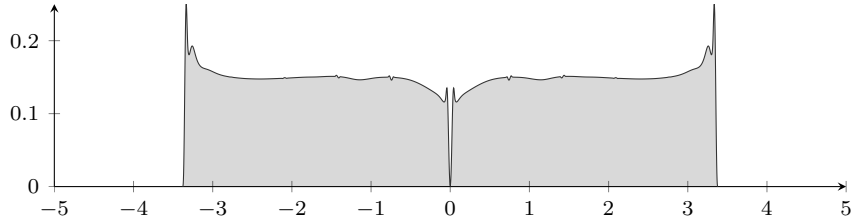

We will only need part (a) of the following lemma, we only give part (b) for
the sake of completeness.

\begin{Lemma}
\cite{god3,hei} \label{pathcover} (a) The maximum multiplicity of a zero of
$\mu(G,x)$ is at most  the number of vertex-disjoint paths required to
cover $G$. \medskip

\noindent(b) The number of distinct zeros of $\mu(G,x)$ is at least the length of 
the longest path in $G$.
\end{Lemma}

The following lemma is a deep result of C. Y. Ku and W. Chen \cite{ku}.

\begin{Lemma}
\label{transitive} \cite{ku} If $G$ is a finite connected vertex transitive
graph, then all zeros of the matching polynomial are distinct.
\end{Lemma}

Now we are ready to give a generalization of Theorem~\ref{atomless}.

\begin{Th}
\label{atomless-general} Let $L$ be a lattice satisfying one of the following
conditions. \medskip

\noindent(a) The lattice $L$ can be obtained as a Benjamini--Schramm limit of
a finite graph sequence $G_{n}$ such that $G_{n}$ can be covered by
$o(|G_{n}|)$ disjoint paths. \medskip

\noindent(b) The lattice $L$ can be obtained as a Benjamini--Schramm limit of
connected vertex transitive finite graphs.
\medskip

Then the matching measure $\rho_{L}$ is atomless.
\end{Th}

\begin{proof}
We prove the two statements together. Let $\mbox{mult}(G_{n},\theta)$ denote
the multiplicity of $\theta$ as a zero of $\mu(G_{n},x)$. Then by
Theorem~\ref{Luck} we have
\[
\rho_{L}(\{\theta\})=\lim_{n\rightarrow\infty}\frac{\mbox{mult}(G_{n},\theta
)}{|G_{n}|}.
\]
Note that by Lemma~\ref{pathcover} we have $\mbox{mult}(G_{n},\theta)$ is at
most the number of paths required to cover the graph $G_{n}$. In case of
connected vertex transitive graphs $G_{n}$, we have $\mbox{mult}(G_{n}%
,\theta)=1$ by Lemma~\ref{transitive}. This means that in both cases $\rho
_{L}(\{\theta\})=0$.
\end{proof}

\begin{proof}
[Proof of Theorem~\ref{atomless}]Note that $\mathbb{Z}^{d}$ satisfies both
conditions of Theorem~\ref{atomless-general} by taking boxes or using part (b),
taking toroidal boxes.
\end{proof}

\begin{center}

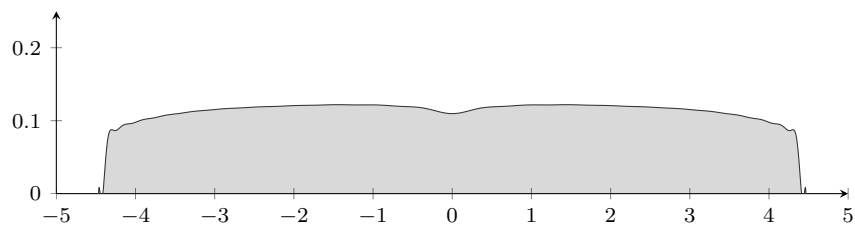
\begin{figure}[h!]
\tikzsetnextfilename{3dgridmm} \begin{tikzpicture}
\begin{axis}[small, height=4cm, width=12cm, xmin=-5, xmax=5, ymin=0, ymax=0.25, axis lines=left, /pgf/number format/fixed]
\addplot[black!80, fill=black, fill opacity=0.15] table {3dgridmm.dat};
\end{axis}
\end{tikzpicture}
\caption{An approximation for the matching measure of $\mathbb{Z}^{3}$.
Working with reasonably sized finite grids would have been computationally too
expensive, so this time we took the $L_{2}$ projection of the infinite measure
to the space of degree 48 polynomials which can be calculated from the
sequence of moments.}%
\label{fig:3dgridmm}%
\end{figure}
\end{center}


\begin{thebibliography}{99}                                                                                               %




\bibitem {ACFK}M. Ab\'ert, P.\ Csikv\'ari, P. E.\ Frenkel and G. Kun:
\textit{Matchings in Benjamini--Schramm convergent graph sequences}, ArXiv preprint 1405.3271, to appear in Trans. Amer. Math. Soc.

\bibitem {xabert_hubai}M. Ab\'ert and T. Hubai: \textit{Benjamini--Schramm
convergence and the distribution of chromatic roots for sparse graphs},
arXiv:1201.3861v1, to appear in Combinatorica

\bibitem{luckapprox}M. Ab\'{e}rt, A. Thom and B. Vir\'{a}g: \textit{
Benjamini-Schramm convergence and pointwise convergence of the spectral
measure}, preprint at \url{www.renyi.hu/~abert}

\bibitem {alm}S. E. Alm:\textit{Upper bounds for the connective constant of
self-avoiding walks}, Comb. Probab. Comp. \textbf{2}, pp. 115--136


\bibitem {bax}R. J. Baxter: \textit{Dimers on a rectangular lattice}, J. Math
Phys. \textbf{9} (1968), pp. 650--654

\bibitem {bupe}P. Butera and M. Pernici: \textit{Yang-Lee edge singularities
from extended activity expansions of the dimer density for bipartite lattices
of dimensionality $2\leq d\leq7$}, ArXiv preprint 1206.0872



\bibitem {csi}P.\ Csikv\'ari and P. E.\ Frenkel: \textit{Benjamini--Schramm
continuity of root moments of graph polynomials}, ArXiv preprint 1204.0463

\bibitem {dar}J. N. Darroch: \textit{On the distribution of the number of
successes in independent trials}, Ann. Math. Statist. \textbf{35}, pp. 1317--1321

\bibitem {dum}H. Duminil-Copin and S. Smirnov: \textit{The connective constant
of the honeycomb lattice equals $\sqrt{2+\sqrt{2}}$}, Ann. of Math. (2)
\textbf{175}(3), pp. 1653--1665

\bibitem {ele}G.\ Elek and G.\ Lippner: \textit{Borel oracles. An analytical
approach to constant-time algorithms}, Proc.\ Amer.\ Math.\ Soc.\ \textbf{138}%
(8) (2010), pp. 2939--2947.

\bibitem {fis}M.\ E.\ Fisher: \textit{Statistical mechanics of dimers on a
plane lattice}, Phys. Rev. \textbf{124} (1961), pp. 1664--1672


\bibitem {friegurv}S. Friedland and L. Gurvits, Lower bounds for partial
matchings in regular bipartite graphs and applications to the monomer-dimer
entropy, Combinatorics, Probability and Computing 17 (2008), 347-361.


\bibitem {FP}S. Friedland and U. N. Peled: \textit{Theory of Computation of
Multidimensional Entropy with an Application to the Monomer-Dimer Problem},
Advances of Applied Math. \textbf{34} (2005), pp. 486--522

\bibitem {gam}D. Gamarnik and D. Katz: \textit{Sequential cavity method for
computing free energy and surface pressure}, J Stat Phys \textbf{137} (2009),
pp. 205--232

\bibitem {god3}C. D. Godsil: \textit{Algebraic Combinatorics}, Chapman and
Hall, New York 1993



\bibitem {gur2}L. Gurvits: \textit{Unleashing the power of Schrijver's
permanental inequality with the help of the Bethe Approximation}, ArXiv
preprint 1106.2844v11

\bibitem {HSS}T. Hara, G. Slade and A. D. Sokal: \textit{New lower bounds on
the self-avoiding-walk connective constant}, J. Statist. Phys. \textbf{72}
(1993), pp. 479--517

\bibitem {ham1}J. M. Hammersley: \textit{Existence theorems and Monte Carlo
methods for the monomer-dimer problem.} In: David, F.N. (ed.) Research Papers
in Statistics: Festschrift for J. Neyman, pp. 125--146. Wiley, London (1966)

\bibitem {ham2}J. M. Hammersley: \textit{An improved lower bound for the
multidimensional dimer problem.} Proc. Camb. Philos. Soc. \textbf{64} (1966),
pp. 455--463

\bibitem {ham3}J. M. Hammersley, V. Menon: \textit{A lower bound for the
monomer-dimer problem.} J. Inst. Math. Appl. \textbf{6} (1970), pp. 341--364

\bibitem {hei}O. J. Heilmann and E. H. Lieb: \textit{Theory of monomer-dimer
systems}, Commun. Math. Physics \textbf{25} (1972), pp. 190--232

\bibitem {huo}Y. Huo, H. Liang, S. Q. Liu, F. Bai: \textit{Computing the
monomer-dimer systems through matrix permanent}, Phys. Rev. E \textbf{77} (2008)

\bibitem {kas}P. W. Kasteleyn: \textit{The statistics of dimers on a lattice,
I: the number of dimer arrangements on a quadratic lattice}, Physica
\textbf{27} (1961), pp. 1209--1225

\bibitem {ku}C. Y. Ku and W. Chen: \textit{An analogue of the Gallai--Edmonds
Structure Theorem for non-zero roots of the matching polynomial}, Journal of
Combinatorial Theory, Series B \textbf{100} (2010), pp. 119--127



\bibitem {ngu}H.\ N.\ Nguyen and K.\ Onak, \textit{Constant-time approximation
algorithms via local improvements}, 49th Annual IEEE Symposium on Foundations
of Computer Science (2008), pp. 327--336.




\bibitem {tem}H.\ N.\ V.\ Temperley and M.\ E.\ Fisher: \textit{Dimer problem
in statistical mechanics--an exact result}, Philos. Mag. \textbf{6} (1961),
pp. 1061--1063

\bibitem {thom} A. Thom: \textit{Sofic groups and diophantine approximation}, Comm. Pure Appl. Math., Vol. LXI, (2008), 1155--1171




\end{thebibliography}
\end{document}